\documentclass[a4paper,USenglish,11pt]{article}
\usepackage{geometry}
\usepackage{graphicx}
\usepackage{etoolbox}
\usepackage{multirow}
\usepackage{booktabs}
\usepackage{graphicx}
\usepackage{bbm}
\usepackage{tabularx,environ}
\usepackage[utf8]{inputenc}
\usepackage{ stmaryrd }
\usepackage{amsthm}
\usepackage{xspace}

\usepackage{booktabs,siunitx}
\usepackage{algorithm}
\usepackage[noend]{algpseudocode}
\usepackage{xcolor,makecell}
\usepackage{tabularx,environ}
\usepackage{xspace}
\usepackage{amsfonts,amsmath}
\usepackage{subcaption}
\usepackage{todonotes}
\usepackage[switch,modulo]{lineno}
\usepackage[numbers]{natbib}
\usepackage{hyperref}
\usepackage{cleveref}

\usepackage{geometry}
\usepackage{tikz}
\usetikzlibrary{arrows}
\usetikzlibrary{calc, fit}
\newtheorem{theorem}{Theorem}
\newtheorem{corollary}{Corollary}

\newcommand{\decprob}[3]{%
  \begin{center}%
    \begin{minipage}{0.9\linewidth}%
      \textsc{#1}\\
      \textbf{Input:} #2\\
      \textbf{Question:} #3
    \end{minipage}%
  \end{center}%
}

\newcommand{\scd}{\textsc{Set Cover with Demands}\xspace}
\newcommand{\scb}{\textsc{Set Cover with Capacities}\xspace}

\title{Finding Small Multi-Demand Set Covers with Ubiquitous Elements and Large 
Sets is Fixed-Parameter Tractable}

\author{
   Niclas Boehmer$^\mathbf{1}$\and Robert Bredereck$^\mathbf{2}$ \and Du{\v
s}an Knop$^\mathbf{3}$ \and
   Junjie
   Luo$^\mathbf{4}$
}
\date{$^1$~Algorithmics and
	Computational Complexity, TU
Berlin, Germany, niclas.boehmer@tu-berlin.de\\
	$^2$~Humboldt-Universität zu Berlin, Germany,
robert.bredereck@hu-berlin.de\\
$^3$~Czech Technical University in Prague, Czech Republic,
dusan.knop@fit.cvut.cz\\
$^4$~Nanyang Technological University, Singapore, junjie.luo@ntu.edu.sg
}

\begin{document}

\maketitle

\begin{abstract}
We study a variant of \textsc{Set Cover} where each element of the
universe has some demand that determines how many times the element
needs to be covered. Moreover, we examine two generalizations of this problem
where a set can be included multiple times and where sets have different prices.
We prove that all three problems are fixed-parameter tractable with
respect to the combined parameter budget, the maximum number of elements
missing in one of the
sets, and the maximum number of sets in which one of the elements does not
occur. Lastly, we point out how our fixed-parameter tractability algorithm can
be applied in the context of bribery for the (collective-decision) group
identification problem.
\end{abstract}

\noindent We consider the following variant of the traditional \textsc{Set
Cover}
problem:
\decprob{Set Cover with Demands}{Universe $U=[n]$, a list of demands
$d_1,\dots , d_n\in [m]$, a family of subsets
$\mathcal{F}=\{F_1,\dots,F_m\}$ over~$U$, and an integer $k\in
\mathbb{N}$.}{Does there exist a subset $\mathcal{S}\subseteq \mathcal{F}$
with $|\mathcal{S}|=k$ such that each element $i\in U$ is included in at least
$d_i$ sets from $\mathcal{S}$?}

We start by introducing some notation.
A family of subsets over some universe~$U$ is called covering system.
Given some covering system~$\mathcal{F}$ and an element $i\in U$, we denote by
$\mathcal{F}(i)$ the subfamily of sets containing~$i$.
Analogously, for a subset of elements $F\subseteq U$, we denote the subfamily
of sets from $\mathcal{F}$ containing any element from $F$ as $\mathcal{F}(F)$, i.e.,
$\mathcal{F}(F)=\bigcup_{i\in F} \mathcal{F}(i)$.

Let $s_{\text{min}}=\min_{i\in m} |F_i|$ be the minimum size of a set from the
covering system and $o_{\text{min}}=\min_{i\in n} |\mathcal{F}(i)|$ the minimum
number of occurrences of an element in sets from the covering system.
Symmetrically, let
$s_{\text{max}}=\max_{i\in m} |F_i|$ be the maximum size of a set from the
covering system and $o_{\text{max}}=\max_{i\in n} |\mathcal{F}(i)|$ the maximum
number of occurrences of an element in sets from the covering system.

We consider \scd parameterized by the number of sets to be selected, the
maximum number of elements missing in a set, and the
maximum number of sets in which an element does not appear. This means
that we develop an algorithm for situations in which the budget is
small, each set contains
almost all elements, and each element is contained in almost all sets.
In the following, we show that \scd parameterized by
$k+(n-s_{\text{min}})+(m-o_{\text{min}})$ is
fixed-parameter tractable.
To this end, we show fixed-parameter tractability of the ``complement'' problem:

\decprob{Set Cover with Capacities}{Universe $U=[n]$, a list of capacities
$c_1,\dots , c_n\in [m]$, a family of subsets
$\mathcal{F}=\{F_1,\dots,F_m\}$ over~$U$, and an integer $k\in
\mathbb{N}$.}{Does there exist a subset $\mathcal{S}\subseteq \mathcal{F}$
with $|\mathcal{S}|=k$ such that each element $i\in U$ is included in at most
$c_i$ sets from $\mathcal{S}$?}

Note that each instance $\mathcal{I}'$ of \scd with parameter $s'_{\text{min}}$
and
$o'_{\text{min}}$ can be easily transformed into an instance $\mathcal{I}$ of
\scb by replacing each set from the covering system by its complement and
setting
$c_i=k-d_i$ for all $i\in U$. By replacing all sets by their complement, it
holds that $s_{\text{max}}=n-s'_{\text{min}}$ and
$o_{\text{max}}=m-o'_{\text{min}}$. Thus, to show our initial claim, it is
enough to show the following theorem.

\begin{theorem} \label{th:1}
 Parameterized by $k+s_{\text{max}}+o_{\text{max}}$, \scb is solvable
in $\mathcal{O}\Big((n+m)\cdot \big(s_{\text{max}}\cdot
o_{\text{max}}\big)^k\Big)$ time.
\end{theorem}
\begin{proof}
    The algorithm for this problem follows a branch-and-bound approach and
constructs a solution $\mathcal{S}$ by adding a set to the solution at each
level. To do so,
the algorithm keeps track of the elements $U'$ that still have free capacity
left and the collection of sets $\mathcal{F}'$ where all elements in the set
have still free capacity. These are the sets that can still be added to
$\mathcal{S}$. At each level, we select one set $F^*$ from
$\mathcal{F}'$ and branch over adding to $\mathcal{S}$ this set or a set that
overlaps with $F^*$ in at least one element. We will see that if there exists
a solution not using any additional of the sets overlapping
with $F^*$, then there needs to exist a solution containing $F^*$, as it is
possible to replace one of the sets in the solution without $F^*$ by $F^*$
without violating the capacity constraints of any element.

    This reasoning gives rise to the recursive algorithm presented in
 	\Cref{alg2}, which can be called to solve the problem as
 	CalcCover($U$, $\mathcal{F}$, $U'$, $\mathcal{F}'$,
$\emptyset$, $k$, $c_1$, \dots, $c_m$).
 	\begin{algorithm}[tb]
 		\caption{CalcCover($U$, $\mathcal{F}$, $U'$,
$\mathcal{F}'$, $\mathcal{S}$, $k$, $c_1$, \dots, $c_m$)}
 		\label{alg:algorithm}
 		\begin{algorithmic}[1] 
 		\Require Univserse $U$, covering system $\mathcal{F}$, elements with
free
capacity $U'$, sets being in question for selection $\mathcal{F}'$, constructed
solution $\mathcal{S}$, and budget $k$.
 			\If {$k=0$}
 			\State Return $\mathcal{S}$
 			\EndIf
 			\While {there exists an $i \in U'$ such that
$|\mathcal{S}\cap \mathcal{F}(i)|= c_i$}
 			\State $\mathcal{F}'=\mathcal{F}'\setminus \mathcal{F}(i)$
 			\State $U'=U'\setminus \{i\}$
 			\EndWhile
 			\If {$\mathcal{F}'=\emptyset$}
 			\State Reject
 			\EndIf
 			\State Let $F^*\in \mathcal{F}'$ be an arbitrary set
from $\mathcal{F}'$
 			\For{$F\in \{F^*\}\cup\big(\mathcal{F}' \cap
 			\mathcal{F}(F^*)\big)$}
             \State Return CalcCover($U$, $\mathcal{F}$, $U'$,
$\mathcal{F}'\setminus \{F\}$, $\mathcal{S}\cup \{F\}$, $k-1$, $c_1$, \dots, $c_m$)
             \EndFor
 	 		\end{algorithmic} \label{alg2}
 	\end{algorithm}

     Every solution returned by the algorithm is also an actual
     solution of the problem, as we only add sets where all elements in the
set still have capacity left. Let $Z\subseteq 2^{\mathcal{F}}$ be the
set of all
     solutions. In the following, assuming that $Z\neq \emptyset$, we
     prove that \Cref{alg2} calculates a solution by proving via
     induction that at each depth $i\in[1,k]$ there exists at least one
     branch which has calculated an $i$-subset of a solution from $Z$.

     We start by examining the initial call of the algorithm, i.e. depth $1$.
     Let
     $F^*$ and $\mathcal{F}'$ be the values of the respective
     variables after line
     8.
     We claim that there exists a solution $\mathcal{Z}\in Z$ of which either
     the selected set $F^*$ or at least one set $F\in \mathcal{F}'\cap
    \mathcal{F}(F^*)$ is part of. Let $\mathcal{Z}'\in Z$ be a solution where
    this is not the case. As $F^*\in \mathcal{F}'$, it
    holds that $c_i>0$ for all $i\in F^*$. Moreover, as no set from
$\mathcal{F}\setminus \mathcal{F}'$ can be part of $\mathcal{Z}'$ and neither
$F^*$ nor a set from $\mathcal{F}'\cap
\mathcal{F}(F^*)$ is part of $\mathcal{Z}'$, all elements from $F^*$ are not
covered at all by $\mathcal{Z}'$. Consequently, it is
     possible to replace an arbitrary set in $\mathcal{Z}'$ by $F^*$.

  	As the induction hypothesis, we assume that at depth $i-1$ there exists at
  	least one branch for which $\mathcal{S}$ is a
    subset of
  	a valid solution, i.e., $\exists \mathcal{Z}\in Z: \mathcal{S}\subseteq
\mathcal{Z}$. Let
  	$F^*$, $\mathcal{S}$, and $\mathcal{F}'$ be the values of the respective
variables after line
  	8 of this branch.  We claim that there always exists an $ F\in
  	\{F^*\}\cup\big(\mathcal{F}' \cap \mathcal{F}(F^*)\big)$ and an
    $\mathcal{Z}\in Z$
  	such that $\mathcal{S}\cup \{F\}\subseteq \mathcal{Z}$.  Let
$\mathcal{Z}'\in Z$ with
  	$\mathcal{S}\subseteq \mathcal{Z}'$ be a solution where neither $F^*$ nor
some $F\in \mathcal{F}' \cap \mathcal{F}(F^*)$ is part of. As $F^*\in
\mathcal{F}'$, it
  	holds that $|\mathcal{S}\cap \mathcal{F}(i)|<c_i$ for all $i\in F^*$.
  	Moreover, as none of $ F\in
  	\{F^*\}\cup\big(\mathcal{F}' \cap \mathcal{F}(F^*)\big)$ is part of
$\mathcal{Z}'$ (by our assumption) and none of
$\mathcal{F}\setminus \mathcal{F'}$ can be part of $\mathcal{Z}'\setminus
\mathcal{S}$ (as this would exceed the capacity of at least one element), it
  	holds that
  	$|\mathcal{Z}'\cap \mathcal{F}(i)|<c_i$ for all $i\in F^*$. Consequently,
it is
  	possible to replace an arbitrary set in $\mathcal{Z}'$ by $F^*$.

  	Concerning the running time, note that the depth of the recursion is
bounded by $k$. As each set contains at most $s_{\text{max}}$ elements and each
element is contained in at most $o_{\text{max}}$ different sets, it holds that
the branching factor (in line 9) is bounded by $s_{\text{max}}\cdot
o_{\text{max}}$. Each step only requires linear time which results in an overall
running time of $\mathcal{O}\Big((n\cdot m)\cdot \big(s_{\text{max}}\cdot
o_{\text{max}}\big)^k\Big)$.
\end{proof}
As argued above, from this, fixed-parameter tractability of our original
problem follows:
\begin{corollary}
 Parameterized by $k+(n-s_{\text{min}})+(m-o_{\text{min}})$, \scd is solvable
in $\mathcal{O}\Big((n\cdot m)\cdot \big((n-s_{\text{min}})\cdot
(m-o_{\text{min}})\big)^k\Big)$ time.
\end{corollary}
In the following, we will explain how the algorithm from above can be adapted
to also work for two variants of the \scd and \scb problems.

However, before we do so, we want to remark that if the only goal is to prove
that the two problems are fixed-parameter tractable with respect to the three
considered parameters and the exact running time of the algorithm is not
important, it is also possible to employ the following simpler algorithm for
\scb.

We start by deleting all elements with zero capacity from the instance and all
sets that contain one of these elements. For the resulting instance, we
distinguish two different cases: In case that the size $m$ of the covering
system $\mathcal{F}$  is smaller than or equal to $k\cdot s_{\text{max}} \cdot
o_{\text{max}}$, we brute force over all $k$-subsets of $\mathcal{F}$ and check
for each of them whether it is a valid solution. Otherwise, it holds that
$m>k\cdot s_{\text{max}} \cdot o_{\text{max}}$. This directly implies that it is
always possible to construct a valid solution as follows. We start by picking
an arbitrary set from $\mathcal{F}$ and include it in the solution
$\mathcal{S}$. Subsequently, we delete all sets from $\mathcal{F}$ in which an
element $i$ occurs for which the number of sets including $i$ in $\mathcal{S}$
is equal to $c_i$. We repeat the two steps from above until $\mathcal{S}$
contains $k$ sets. As each set contains at most $s_{\text{max}}$ elements, in
each step, only the capacity of at most $s_{\text{max}}$ different elements can
get full. Moreover, as each element is contained in at most $o_{\text{max}}$
different sets, in each step, only at most  $s_{\text{max}}\cdot o_{\text{max}}$
sets can get deleted from the covering system $\mathcal{F}$. Since it holds that
$m>k\cdot s_{\text{max}} \cdot
o_{\text{max}}$, it is always possible to construct a solution of size $k$, as
there will always be a set left in $\mathcal{F}$ that one can pick during the
construction of the solution.

Note that this algorithm has a running time of $\mathcal{O}(2^{k\cdot
s_{\text{max}} \cdot
o_{\text{max}}})$. The
approach can be extended if multiplicities are allowed but does not longer work
in the presence of prices.

\subsection*{Introducing Multiplicities}

\noindent So far, we required that each set
from the
covering system can be only included once in the solution. However, it is also
possible to allow that a set is allowed to be included an arbitrary number of
times in the solution arriving at the following adapted versions of our two
problems:
\decprob{Set Cover with Demands (Capacities) and Multiplicities}{Universe
$U=[n]$, a list of demands
	$d_1,\dots , d_n\in [m]$ (capacities
	$c_1,\dots , c_n\in [m]$), a family of subsets
	$\mathcal{F}=\{F_1,\dots,F_m\}$ over~$U$, and an integer~$k\in
	\mathbb{N}$.}{Do there exist integer multiplicities $\ell_1+\dots +
	\ell_m=k$ such that for each element $i\in U$ it holds that $\sum_{j\in
	\mathcal{F}(i)} \ell_j$ is at
	least
	$d_i$ (at most $c_i$)?}
\Cref{alg2} needs to be only slightly adapted to solve these problems.
Firstly,  $\mathcal{S}$ needs to be a multiset of sets and secondly, line 10 of
\Cref{alg2} needs to be modified as follows.
Instead of excluding the selected set $F$ from the collection of sets
$\mathcal{F}'$ that are still in question to be used, $F$ remains in
$\mathcal{F}'$. Thereby, each set can be selected multiple times. In fact, all
arguments from the proof of \Cref{th:1} still apply here. The only small
modification that needs to be made is that in the induction step we assume that
$\mathcal{Z}'\in Z$ with
$\mathcal{S}\subseteq \mathcal{Z}'$ is a solution where neither $F^*$ nor
some $F\in \mathcal{F}' \cap \mathcal{F}(F^*)$ is part of
$\mathcal{Z'}\setminus \mathcal{S}$. This proves the
following corollary:
\begin{corollary}
	Parameterized by $k+s_{\text{max}}+o_{\text{max}}$, \textsc{Set Cover with
	Capacities and Multiplicities} is solvable
	in $\mathcal{O}\Big((n\cdot m)\cdot \big(s_{\text{max}}\cdot
	o_{\text{max}}\big)^k\Big)$ time.\\
	Parameterized by $k+(n-s_{\text{min}})+(m-o_{\text{min}})$, \textsc{Set
	Cover with Demands and Multiplicities} is solvable
	in $\mathcal{O}\Big((n\cdot m)\cdot \big((n-s_{\text{min}})\cdot
	(m-o_{\text{min}})\big)^k\Big)$ time.
\end{corollary}

\subsection*{Introducing Prices}

It is also possible to consider a generalized version of the two
considered
problems where the different sets have different prices:
\decprob{Set Cover with Demands (Capacities) and Prices}{Universe $U=[n]$, list
of demands
	$d_1,\dots , d_n\in [m]$ (capacities $c_1,\dots , c_n\in [m]$), a list of
	prices $p_1, \dots , p_n\in \mathbb{N}$, a family of subsets
	$\mathcal{F}=\{F_1,\dots,F_m\}$ over~$U$, and two integers $k\in
	\mathbb{N}$ and $t\in \mathbb{N}$.}{Does there exist a subset
	$\mathcal{S}\subseteq \mathcal{F}$
	with $|\mathcal{S}|=k$ and $\sum_{F_i \in \mathcal{S}} p_i\leq t$ such that
	each element $i\in U$ is included in at
	least
	$d_i$ (at most $c_i$) sets from~$\mathcal{S}$?}
Note that, in principle, it is also possible to drop the constraint that
exactly $k$ sets need to be selected here. However, the resulting variant with
capacities would become
trivial then. Because of this and to ensure that the two problems can be still
directly
related, we selected the formulation from above.   Nevertheless, our algorithm
for the problems from above is also applicable to different variants of the
problems without $k$, since as a first step it is possible to guess the
value of $k$ which is guaranteed to lie between $1$ and $t$.

\Cref{alg2} can be again slightly adjusted to solve \textsc{Set Cover with
Capacities and Prices}, and, thus, also the problem with demands. To do so, we
need to provide the prices of the sets and the budget $t$ as part of the input
of the algorithm. Moreover, in line 6, we also reject if the price of the
constructed solution $\mathcal{S}$ exceeds~$t$. Lastly, in line 8, we always
pick the set with the lowest price from $\mathcal{F}'$. Again the argumentation
from the proof of \Cref{th:1} still applies here. The only additional
observation one needs is that $F^*$ is always the cheapest set from
$\mathcal{F}'$. Thus, if there exists a solution $\mathcal{Z}'$ with
$\mathcal{S}\subseteq \mathcal{Z}'$  not containing $F^*$ or some set from
$\mathcal{F}' \cap
\mathcal{F}(F^*)$, it is always possible to replace one set from
$\mathcal{Z}'\setminus \mathcal{S}$ by $F^*$, as all sets from
$\mathcal{Z}'\setminus
\mathcal{S}$ need to be part of $\mathcal{F}'$, all elements from $F^*$ need to
have free capacity and $F^*$ is guaranteed to be
not more expensive than all sets from $\mathcal{Z}'\setminus
\mathcal{S}$.
\begin{corollary}
 	Parameterized by $k+s_{\text{max}}+o_{\text{max}}$, \textsc{Set Cover with
 		Capacities and Prices} is solvable
 	in $\mathcal{O}\Big((n\cdot m)\cdot \big(s_{\text{max}}\cdot
 	o_{\text{max}}\big)^k\Big)$ time.\\
 	Parameterized by $k+(n-s_{\text{min}})+(m-o_{\text{min}})$, \textsc{Set
 		Cover with Demands and Prices} is solvable
 	in $\mathcal{O}\Big((n\cdot m)\cdot \big((n-s_{\text{min}})\cdot
 	(m-o_{\text{min}})\big)^k\Big)$ time.
\end{corollary}

\section*{Application: Bribery in Group Identification}
We now describe an application where the \textsc{Set Cover with Capacities}
problem naturally arises and our FPT algorithm is directly applicable: In group
identification, we are given a set $A=\{a_1,\dots,a_n\}$ of agents and the
task
is to identify a
so-called \emph{socially qualified} subgroup of the agents
\cite{DBLP:journals/LeA/KasherR1997}. To do so, we are
given a \textit{qualification profile}~\mbox{$\varphi \colon A \times A \to
\{-1,1\}$} that denotes for each agent $a$ which of the other agents $a$ deems
qualified, i.e., agent~$a$
\textit{qualifies}~$a'$ if~$\varphi(a,a')=1$ and
\textit{disqualifies}~$a'$ if~$\varphi(a,a')=-1$. For an agent~$a\in A$,
let~$Q_\varphi^+(a)=\{a'\in A\mid \varphi(a',a)=1\}$ denote the set of agents
qualifying~$a$ and $Q_\varphi^-(a)=\{a'\in A\mid \varphi(a',a)=-1\}$
the set
of
agents disqualifying~$a$.
To decide given a set of agents and a qualification profile which agents are
socially qualified, different social rules have been proposed. One popular rule
parameterized by two integers $s$ and $t$ with~$s+t\leq n+2$ is the consent
rule, denoted
\ensuremath{f^{(s,t)}}
\cite{DBLP:journals/jet/SametS03}. Under the consent rule, an agent $a\in A$
with $\varphi(a,a)=1$ is socially qualified if and only if at least $s$ agents
(including $a$ itself) qualify $a$. Similarly, an agent $a\in A$ with
$\varphi(a,a)=-1$ is socially disqualified if and only if at least $t$ agents
(including $a$ itself) disqualify $a$.

Recently, \citet{DBLP:journals/aamas/ErdelyiRY20} initiated the study of the
computational complexity of
bribery in the context of group identification, among others, asking the
following question: 
\decprob{Constructive-\ensuremath{f^{(s,t)}} Agent Bribery}{Set $A$ of agents,
qualification profile $\varphi$, subset $A^+\subseteq A$ of agents to be made
socially qualified, and budget
$\ell$.}{Is it possible to modify the opinion of at most $\ell$ agents such
that
after the modifications all agents from $A^+$ are socially qualified under the
consent rule
\ensuremath{f^{(s,t)}}?}
\citet{DBLP:journals/aamas/ErdelyiRY20} proved that this problem is NP-hard
even for $s=1$ and $t=2$. Subsequently, \citet{DBLP:conf/ijcai/BoehmerBKL20}
conducted a detailed study of the
parameterized complexity of this question considering the parameters $\ell$,
$s$, $t$, and $|A^+|$. Among others, they proved that
\textsc{Constructive-\ensuremath{f^{(s,t)}} Agent Bribery} is W[1]-hard with
respect to $s+\ell$ even if $t=1$.

Let $\Delta$ be the maximum number of
agents  from $A^+$ that an agent $a\in A$ qualifies, i.e., $\Delta:=\max_{a\in
A}
|\{a'\in A^+\mid \varphi(a,a')=1\}|$.
We now prove that our algorithm for \textsc{Set Cover with Demands} can be used
to prove that \textsc{Constructive-\ensuremath{f^{(s,t)}} Agent Bribery} with
$t=1$ is
fixed-parameter tractable with respect to $\Delta+s$.

\begin{theorem}\label{th:2}
\textsc{Constructive-\ensuremath{f^{(s,t)}} Agent Bribery} is solvable in
$\mathcal{O}\big(n^2\cdot (\Delta \cdot
s)^s\big)$ time for $t=1$.
\end{theorem}
\begin{proof}
	First of all, as $t=1$ implies that every agent that disqualifies itself is
	also socially disqualified, we bribe all agents in $A^+$ who do not qualify
	themselves to
	qualify everyone and adjust the budget $\ell$ accordingly. We delete from
	$A^+$ all agents who are already socially qualified after this bribery,
	while keeping them in the set of agents.

	If we have $\ell\geq s$ for the resulting budget, we are done as we can
	simply pick $\ell$ agents and make them qualify everyone, which results in
	all agents from $A^+$ being socially qualified.

	Consequently, we are left with the situation where $\ell < s$. We now
	reduce the problem to an instance of \textsc{Set Cover with Demands} as
	follows. We set the universe $U=A^+$ and for each $a\in A^+$ its demand to
	$s-|Q^+(a)|$ (the number of additional qualification $a$ needs to get by
	the bribery to become socially qualified). For each agent $a\in A$ who does
	not qualify all other agents, we add
	a set $F_a$ to our covering
	system $\mathcal{F}$ containing all agents from $A^+$ which $a$ does not
	qualify, i.e.,
	$F_a:=\{a'\in A^+\mid \varphi(a,a')=-1\}$. Finally, we set $k:=\ell$.
	Bribing an agent $a\in A$, which results in all agents from $F_a$ getting
	an
	additional qualification, corresponds to including $F_a$ in the cover. It
	is easy to see that there exists a successful bribery if and only if there
	exists a solution to the constructed \textsc{Set Cover with Demands}
	instance. Note that in the constructed instance $k$ is bounded by $s$.
	Moreover, for each $a\in A$, it holds that $n-|F_a|$ is equal to the number
	of agents $a$ qualifies before the bribery and, thus,
	$\Delta\geq |A^+|-\min_{F\in \mathcal{F}}|F|$. Lastly, note that as each
	agent $a\in
	A^+$ can only be approved by at most $s-1$ agents before the bribery, $a$
	needs to appear in
	all but at most $s-1$ sets. Thus, applying the algorithm from \Cref{th:1},
	we can solve the problem in $\mathcal{O}\big(n\cdot (\Delta \cdot
	s)^s\big)$ time.
\end{proof}

It is even possible to extend this result to a fixed-parameter tractable
algorithm for the parameters $s+t+\ell+\Delta$ (note that as proven by
\citet{DBLP:conf/ijcai/BoehmerBKL20}
\textsc{Constructive-\ensuremath{f^{(s,t)}} Agent Bribery} is W[1]-hard with
respect to $s+t+\ell$):
\begin{corollary}
	\textsc{Constructive-\ensuremath{f^{(s,t)}} Agent Bribery} is solvable in
	$\mathcal{O}\Big(n^2\cdot \big( (\ell+t)^\ell \cdot (\Delta \cdot
	s)^s \big)\Big)$ time.
\end{corollary}

\begin{proof}
	If we bribe an agent, we always make him qualify all agents.
	For all agents~$a \in A^+$ with~$\varphi(a,a)=-1$ and~$|Q^-(a)| \geq
	\ell+t$, $a$
	must be bribed; so we bribe $a$.
	Thus, we can assume~$|Q^-(a)| < \ell+t$ for all~$a \in A^+$
	with~$\varphi(a,a)=-1$.
	Now, as long as there exists an $a \in A^+$
	with~$\varphi(a,a)=-1$, we branch on bribing~$a$ or bribing~$|Q^-(a)|-(t-1)$
	agents
	from~$Q^-(a)$ and update $\varphi$, $\ell$, and $A^+$ accordingly (we
	delete agents
	from $A^+$ if they became socially qualified). We reject the current branch
	if $\ell<0$. For each non-rejected branch, it remains to
	consider agents
	from
	$A^+$ who qualify themselves. This problem is similar to the case
	when~$t=1$ and we can apply Theorem \ref{th:2}.

	As the branching
	factor
	in each step is bounded by $\ell+t$ and the depth is bounded by $\ell$, the
	algorithm from Theorem \ref{th:2} is employed at most $(\ell+t)^\ell$
	times, which results in an overall running time of
	$\mathcal{O}\Big(n^2\cdot \big( (\ell+t)^\ell \cdot (\Delta \cdot
	s)^s \big)\Big)$
\end{proof}

\subsubsection*{Acknowledgments}
NB is supported by the DFG project MaMu  (NI 369/19).
DK is partly supported by the OP VVV MEYS funded project
CZ.02.1.01/0.0/0.0/16\_019/0000765 ``Research Center for Informatics''.

\bibliographystyle{plainnat}

\end{document}